\documentclass[11pt,a4paper]{article}
\usepackage[T1]{fontenc}
\usepackage[utf8]{inputenc}
\usepackage{lmodern}
\usepackage[margin=25mm,headheight=14pt]{geometry}
\usepackage{amsmath,amssymb,amsthm,bm}
\usepackage{graphicx,booktabs,array,tabularx}
\usepackage{microtype}
\usepackage[numbers,sort&compress]{natbib}

\usepackage{url}
\usepackage[hidelinks,pdfencoding=auto]{hyperref}
\usepackage[font=small,labelfont=bf]{caption}
\usepackage{placeins}
\hypersetup{pdftitle={Beta-Oslo and Shape analyses with unresolved parent-state mixtures},pdfauthor={Yangping Shen},pdfsubject={Statistical nuclear gamma decay, mixed beta-decaying parents, identifiability}}
\newcolumntype{Y}{>{\raggedright\arraybackslash}X}
\newcommand{\T}{\mathcal{T}}
\newcommand{\W}{\mathcal{W}}
\newcommand{\dd}{\mathrm{d}}
\newcommand{\TV}{d_{\mathrm{TV}}}
\newcommand{\KL}{D_{\mathrm{KL}}}
\newcommand{\rb}{^{90}\mathrm{Rb}}
\newcommand{\sr}{^{90}\mathrm{Sr}}

\newtheorem{proposition}{Proposition}

\numberwithin{equation}{section}
\title{\vspace{-1.2cm}$\boldsymbol{\beta}$-Oslo and Shape analyses with\protect\\ unresolved parent-state mixtures}
\author{Yangping Shen\thanks{\href{mailto:ypshen@ciae.ac.cn}{ypshen@ciae.ac.cn}}\\[3pt]
\small Department of Nuclear Physics, China Institute of Atomic Energy,\\[-2pt]
\small Beijing 102413, China}
\date{20 September 2026}
\begin{document}
\maketitle
\begin{abstract}
An unresolved mixture of beta-decaying ground and isomeric states populates a common daughter nucleus through different excitation-energy and spin-parity distributions. A framework is developed to determine when such data admit conventional $\beta$-Oslo and Shape analyses and which quantities remain identifiable. The physical primary-$\gamma$ distribution is an average of separately normalized branching kernels; the global parent fraction does not, in general, equal the fraction in an excitation-energy bin or an accepted primary-$\gamma$ sample. Population invariance, Oslo factorization, and recovery of the total nuclear level density are separated. An exact log-linear compatibility test and a lower-cascade replacement identity isolate factorization and first-generation-subtraction effects. For Shape analysis, matched final-state spins and parities remove common angular-momentum factors but can retain excitation-dependent E1/M1 weights. A graph representation gives consistency conditions and exposes the information lost when only one fixed final-state separation is measured. Deterministic spin-parity-resolved cascade benchmarks reproduce the population-invariant limit and exhibit both exact factorization with an effective, rather than total, level density and small matrix residuals accompanied by appreciable shape errors. Evaluated $\rb\rightarrow\sr$ properties specify a relevant application, including two $2^+$ candidate final states and time-dependent parent mixing. The numerical examples are controlled synthetic tests, not an extraction from $^{90}$Rb data. The results provide criteria for interpreting unresolved-parent measurements and for identifying the additional observations required for population-aware inference.
\end{abstract}

\section{Introduction}\label{sec:intro}
The Oslo method extracts a nuclear level density (NLD) and a $\gamma$-ray transmission coefficient from the excitation-energy dependence of primary-$\gamma$ spectra~\cite{Schiller2000}. Its experimental analysis separates the first transition from subsequent cascade radiation before fitting the normalized primary matrix~\cite{Guttormsen1987}. The $\beta$-Oslo extension uses beta decay to populate the daughter nucleus, thereby making statistical $\gamma$-decay information accessible for nuclei for which conventional reaction measurements are difficult~\cite{Spyrou2014}. The quantities inferred in this way can constrain reaction calculations, provided that their normalization and the relation between the measured and relevant spin-parity ensembles are established.

A beta-decaying sample may contain both a ground state and a long-lived isomer. The daughter isotope can then be unambiguous while the parent state remains unresolved for individual events. Both components probe the same set of daughter levels, but their beta-strength distributions and angular-momentum selection rules need not agree. Even a precisely determined integrated isomeric fraction leaves open how the mixture varies with excitation energy and how much of the daughter level space is sampled. The resulting problem involves both the validity of an analysis model and the information available to determine its parameters.

Population dependence is already a recognized limitation of statistical $\gamma$-decay analysis. Systematic studies of the Oslo method have examined restricted spin ranges, parity imbalance, and the assumptions entering primary-$\gamma$ subtraction~\cite{Larsen2011}. Zeiser et al. developed a forward-simulation correction for restricted-spin population using a reaction-specific application~\cite{Zeiser2019}. An unresolved-parent treatment should therefore build on, rather than replace, the established population-dependent framework. Linear addition of parent contributions is a starting identity; it does not by itself establish unbiased extraction or uniqueness of the inferred NLD and strength.

The Shape method provides complementary constraints through primary transitions to selected discrete final levels~\cite{Wiedeking2021}. Its use with beta decay has also been developed to constrain partial NLDs~\cite{Muecher2023}. Matching the final-state spins and parities is particularly useful because the associated initial-spin selection windows coincide. The role of the initial ensemble nevertheless requires care when E1 and M1 contributions coexist. The different effective strength functions obtained from two parent-state selections in the recent $^{70}$Zn study illustrate the importance of accessible final-state parities; that work interprets the difference without requiring a violation of the generalized Brink--Axel hypothesis~\cite{Ronning2026}.

This paper treats unresolved parent populations as a conditional-probability and inverse problem. Three questions are kept distinct: whether a normalized primary spectrum changes when the mixture changes; whether that spectrum belongs to the conventional Oslo factorized family; and whether a factorization represents the total NLD and an intrinsic strength function. The treatment also separates the true primary matrix from the output of a first-generation estimator. Analytical statements are tested with reproducible, deterministic cascade calculations in which the input physics is known. Evaluated $\rb$ and $\sr$ properties motivate a specific application, but no experimental matrix, detector-response simulation, or empirical $^{90}$Sr NLD or strength extraction is presented.

\section{Statistical \texorpdfstring{$\gamma$}{gamma} decay from mixed parent populations}\label{sec:model}
\subsection{Parent fractions and excitation-energy-dependent feeding}
Let $a$ label the parent state and let $\eta_a$ denote its fraction of beta decays in a specified observation interval and feeding domain, with $\sum_a\eta_a=1$. This fraction is not, in general, an implanted-ion fraction. Let $B_a(E_i)$ be the normalized daughter excitation-energy distribution per such beta decay, and $q_a(s|E_i)$ the conditional probability for the spin-parity class $s=(J_i,\pi_i)$. An incoherent mixture gives the expected number of decays into an interval $\dd E_i$ as
\begin{equation}
\dd N(E_i)=N_\beta\sum_a\eta_a B_a(E_i)\dd E_i.
\end{equation}
The local parent fraction and mixed class distribution are consequently
\begin{align}
w_a(E_i)&=\frac{\eta_a B_a(E_i)}{\sum_b\eta_b B_b(E_i)},\label{eq:localw}\\
q_{\rm mix}(s|E_i)&=\sum_a w_a(E_i)q_a(s|E_i).\label{eq:qmix}
\end{align}
These expressions are defined where the denominator is positive. A weak integrated component can dominate a particular excitation interval if the other parent has little feeding there. Conversely, an uncertainty in $\eta_a$ can be nearly irrelevant in a region populated by only one component.

The $B_a$ in Eq.~\eqref{eq:localw} includes all beta-transition classes retained in the feeding domain. It is not a gamma-cascade multiplicity distribution. If beta efficiencies differ through the endpoint energy, transition shape, or event selection, the accepted beta-feeding distribution must be used instead. An energy-migration matrix generally prevents this correction from being represented by a single multiplicative function of true $E_i$.

\subsection{Microscopic branching and statistical averaging}\label{sec:micro}
Consider first a daughter level $\nu$ with excitation energy in bin $E_i$. For prompt radiative decay with negligible competing channels, its branching probability to a level $f$ is
\begin{equation}
b_{\nu f}=\frac{\Gamma_{\nu f}}{\sum_{f'}\Gamma_{\nu f'}},
\qquad \sum_f b_{\nu f}=1.\label{eq:microbranch}
\end{equation}
If internal conversion or particle emission is relevant, the total width must include those channels and the observable must be conditioned on the corresponding event selection. The central derivations below concern a selected region in which prompt gamma decay supplies the relevant cascade.

Write the parent-dependent level population as
\begin{equation}
p_a(\nu|E_i)=q_a(s|E_i)h_a(\nu|s,E_i),
\qquad \sum_{\nu\in s}h_a(\nu|s,E_i)=1.
\end{equation}
The exact primary-energy kernel for that parent and class is
\begin{equation}
K_{a,s}(E_\gamma|E_i)=
\sum_{\nu\in s}h_a(\nu|s,E_i)
\sum_f b_{\nu f}\delta(E_\gamma-E_\nu+E_f).\label{eq:microkernel}
\end{equation}
Thus, a shared daughter does not automatically imply a shared coarse-grained $K_s$. Different beta operators can populate different levels within the same $E_i,J^\pi$ interval. Replacing $K_{a,s}$ by a parent-independent statistical kernel requires sufficiently representative averaging, or an explicit model of the residual population--decay correlations.

Under that statistical assumption, introduce the unnormalized continuum kernel
\begin{equation}
G_s(E_i,E_\gamma)=
\sum_{XL}\T_{XL}(E_\gamma)
\sum_{J_f,\pi_f}
\rho(E_i-E_\gamma,J_f,\pi_f)
S_{XL}(s;J_f,\pi_f),\label{eq:G}
\end{equation}
where $S_{XL}$ specifies the allowed transition classes and any adopted angular-momentum weights. Discrete final levels are added as individual terms rather than counted again in $\rho$. The convention
\begin{equation}
\T_{XL}(E_\gamma)=2\pi E_\gamma^{2L+1}f_{XL}(E_\gamma)
\end{equation}
is used. A common mean initial-level-spacing factor in the average partial widths cancels in the branching ratio. Define
\begin{equation}
Z_s(E_i)=\int G_s(E_i,E_\gamma)\dd E_\gamma,
\qquad K_s(E_\gamma|E_i)=\frac{G_s(E_i,E_\gamma)}{Z_s(E_i)}.\label{eq:normalizedK}
\end{equation}
Here $Z_s$ is the normalization of a transmission-density kernel; it is not a measured lifetime or a beta-population weight.

Equations~\eqref{eq:G}--\eqref{eq:normalizedK} are a mean-width approximation to Eq.~\eqref{eq:microkernel}. In a finite statistical ensemble, $\langle\Gamma_{\nu f}/\Gamma_\nu\rangle$ need not equal $\langle\Gamma_{\nu f}\rangle/\langle\Gamma_\nu\rangle$. Transition-width fluctuations~\cite{Porter1956} and their consequences for cascade spectra can be treated in level-resolved simulation, as implemented in RAINIER~\cite{Kirsch2018}. They are excluded from the deterministic benchmarks below so that population and normalization effects can be isolated.

\subsection{Order of normalization and mixing}
Under the shared-kernel approximation, the physical primary distribution is
\begin{equation}
\boxed{P_{\rm mix}(E_\gamma|E_i)
=\sum_s q_{\rm mix}(s|E_i)\frac{G_s(E_i,E_\gamma)}{Z_s(E_i)}.}\label{eq:master}
\end{equation}
Each initial class contributes one primary photon per selected decay before any gamma-energy cut. The alternative expression
\begin{equation}
\frac{\sum_s q_{\rm mix}(s|E_i)G_s(E_i,E_\gamma)}
{\sum_s q_{\rm mix}(s|E_i)Z_s(E_i)}\label{eq:wrong}
\end{equation}
instead weights normalized kernels by $q_sZ_s/\sum_tq_tZ_t$. It agrees with Eq.~\eqref{eq:master} only under additional conditions, for example equal $Z_s$ or identical normalized kernels. Without such conditions it changes the ensemble being sampled.

A two-bin example makes the distinction explicit. Let $G_g=(0.9,0.1)$, $G_m=3(0.1,0.9)$, and $q_g=q_m=1/2$. Equation~\eqref{eq:master} gives $(0.5,0.5)$; Eq.~\eqref{eq:wrong} gives $(0.3,0.7)$. The difference arises entirely from the order of normalization, despite the exact knowledge of the mixture. The correct result is recovered by summing unnormalized \emph{event counts} whose per-event decay probabilities have already been properly normalized.

\subsection{Conditioning on the analysis window and detector selection}
For a gamma-energy window $\W_i$, let
\begin{equation}
A_s(E_i)=\int_{\W_i}K_s(E_\gamma|E_i)\dd E_\gamma.
\end{equation}
The accepted class weight is
\begin{equation}
\widetilde q_s(E_i)=\frac{q_s(E_i)A_s(E_i)}{\sum_tq_t(E_i)A_t(E_i)},
\end{equation}
and the accepted normalized spectrum becomes
\begin{equation}
P_{\W}(E_\gamma|E_i)=\sum_s\widetilde q_s(E_i)
\frac{K_s(E_\gamma|E_i)}{A_s(E_i)},\quad E_\gamma\in\W_i.\label{eq:window}
\end{equation}
There is an analogous parent-level expression
\begin{equation}
\widetilde w_a=\frac{w_a A_a}{\sum_b w_b A_b},\qquad
A_a=\int_{\W_i}P_a(E_\gamma|E_i)\dd E_\gamma.\label{eq:acceptedw}
\end{equation}
These tilts are present even for an ideal detector if low-energy primary photons are excluded.

For a fixed, linear detector response without unresolved pileup, expected raw matrices add linearly. This property remains true when each parent is propagated through the response separately. It does not imply that an unfolded, row-normalized, or first-generation-extracted matrix is the same linear combination with the original integrated weights. All selection and migration should be placed at their physical position in the forward chain.

\subsection{A direct measure of mixture sensitivity}
For two fixed component spectra in a specified row and window, define $P(w)=(1-w)P_g+wP_m$. The total-variation distance obeys
\begin{equation}
\TV[P(w),P(v)]=|w-v|\,\TV(P_g,P_m),\qquad
\TV(p,q)=\frac12\int|p-q|\dd E_\gamma.\label{eq:tv}
\end{equation}
The proof follows by subtracting the two mixtures. This identity separates uncertainty in the local mixture from distinguishability of the component decay kernels. It also shows why an isomeric fraction alone cannot define a general tolerance. Equation~\eqref{eq:tv} bounds a change in the observable spectrum; an ill-conditioned inverse problem can amplify that change into a much larger parameter shift.

\section{Validity and identifiability of \texorpdfstring{$\beta$}{beta}-Oslo extraction}\label{sec:oslo}
\subsection{Population invariance and factorization are separate statements}
\begin{proposition}[Population invariance]\label{prop:invariance}
At fixed $E_i$, suppose all significantly populated classes have the same normalized kernel, $K_s=K_0$. Then the primary spectrum is independent of their population weights. Conversely, invariance under every possible probability distribution over those classes implies equality of their kernels.
\end{proposition}
\begin{proof}
The forward implication follows from $\sum_s q_s=1$. For the converse, choose distributions concentrated on each class in turn.
\end{proof}
For variations restricted to two specific parent ensembles, only $P_g=P_m$ is required; equality of every class kernel is a stronger sufficient condition. Accidental cancellations for one parent pair should not be generalized to arbitrary beta populations.

The usual Oslo model is
\begin{equation}
P^{\rm O}(E_\gamma|E_i)=
\frac{\rho(E_i-E_\gamma)\T(E_\gamma)}
{\int_{\W_i}\rho(E_i-E'_\gamma)\T(E'_\gamma)\dd E'_\gamma}.\label{eq:oslo}
\end{equation}
Population invariance does not establish that the common $K_0$ has this form. Conversely, a population-dependent spectrum can belong to this family with effective parameters. In particular, suppose a fixed initial class accesses a fraction $a(E_f)$ of the final-level density and the multipolarity dependence is separable. Then
\begin{equation}
P(E_\gamma|E_i)\propto
\underbrace{\rho_{\rm total}(E_f)a(E_f)}_{\rho_{\rm eff}(E_f)}\T(E_\gamma).\label{eq:partialfactor}
\end{equation}
The factorization is exact even when $a(E_f)$ is strongly energy dependent. A perfect fit therefore cannot establish that $\rho_{\rm eff}$ is the total NLD. This is the same physical concern that motivates population corrections in restricted-spin analysis~\cite{Zeiser2019}; Eq.~\eqref{eq:partialfactor} provides an explicit limiting case for mixed-parent tests.

\subsection{An exact compatibility test for a positive binned matrix}\label{sec:loglinear}
Consider an additive energy grid, with an observed positive normalized matrix $P_{ij}$ on a specified support $\Omega$. Let $f(i,j)$ be the final-energy index. Construct a design matrix $X$ with one row per observed cell and columns for a final-energy parameter $r_f$, a gamma-energy parameter $t_j$, and an initial-row constant $c_i$. Its action is
\begin{equation}
(X\theta)_{ij}=r_{f(i,j)}+t_j+c_i.\label{eq:design}
\end{equation}
\begin{proposition}[Binned factorization compatibility]\label{prop:factorization}
On a positive support with each row normalized over that support, a matrix admits the discrete form of Eq.~\eqref{eq:oslo} if and only if $\log P$ belongs to the column space of $X$.
\end{proposition}
\begin{proof}
Taking logarithms of a factorized matrix gives Eq.~\eqref{eq:design}. In the reverse direction set $\rho_f=e^{r_f}$ and $\T_j=e^{t_j}$. The row sum of $P$ forces $e^{c_i}$ to be the inverse of the corresponding partition sum. A constant bin-width factor is absorbed into $\rho_f$.
\end{proof}
A log-space projection residual $(I-XX^+)\log P$ is therefore an exact algebraic diagnostic for noiseless positive matrices. Missing support can create additional null directions. Zero counts cannot be handled by simply taking logarithms; a count likelihood or a justified regularization is needed. Nor is the squared projection residual automatically a chi-squared statistic, because row normalization and preceding analysis steps induce covariance.

This criterion concerns representability, not nuclear interpretation. It supplies a useful separation between a genuinely nonfactorizable matrix and a factorized matrix whose parameters differ from the intended total quantities. Section~\ref{sec:numerical} exhibits both cases.

\subsection{First-generation subtraction: a separate population condition}\label{sec:fg}
Let $H_{is}(j)$ denote the expected number of photons in gamma bin $j$ emitted by a complete cascade starting in class $s$ of initial bin $i$. It is a photon-yield spectrum, not a unit-normalized probability distribution. The directly populated spectrum is $F_i=\sum_s q_{is}H_{is}$. Let $b_i(f,t)$ be the joint probability that the first transition ends in final-energy bin $f$ and class $t$, after averaging over the initial population. Then
\begin{equation}
F_i=P_i+\sum_{f<i,t}b_i(f,t)H_{ft},\qquad
b_i(f)=\sum_t b_i(f,t).\label{eq:cascade}
\end{equation}
The true primary spectrum $P_i$ is unit normalized when every cascade begins with one tracked photon.

Consider an \emph{oracle replacement} that knows the exact first-destination energy probabilities, but uses the measured, directly populated lower-bin spectrum as its cascade template:
\begin{equation}
\widehat P_i^{\rm or}=F_i-\sum_{f<i}b_i(f)F_f.
\end{equation}
Define the spectrum reached through the first transition as
\begin{equation}
H_f^{\rm cas|i}=\sum_t\frac{b_i(f,t)}{b_i(f)}H_{ft}
\end{equation}
for $b_i(f)>0$. Subtraction of Eq.~\eqref{eq:cascade} gives the exact identity
\begin{equation}
\boxed{\widehat P_i^{\rm or}-P_i=
\sum_{f<i}b_i(f)\bigl[H_f^{\rm cas|i}-F_f\bigr].}\label{eq:oracle}
\end{equation}
Equality of the directly and cascade-populated average spectra is sufficient for this residual to vanish. Different class distributions need not cause a residual if their cascade kernels are identical; cancellations between terms are also possible. The condition is thus about average decay spectra, not merely equality of spin populations.

The usual first-generation procedure estimates subtraction weights iteratively and includes its own normalization conventions~\cite{Guttormsen1987,Larsen2011}. Equation~\eqref{eq:oracle} is a controlled diagnostic of its physical template assumption, not a substitute implementation or a prediction of the bias of a particular software package. An oracle residual can be signed and need not sum to zero. Negative extracted bins are possible before any positivity enforcement. This also establishes why linearity of true event or primary counts does not imply that mixture formation commutes with the complete first-generation estimator.

\subsection{Normalization freedom and unobserved level space}\label{sec:identifiability}
Even an exact conventional factorization retains the familiar invariance~\cite{Schiller2000}
\begin{equation}
\rho'(E_f)=A e^{\alpha E_f}\rho(E_f),\qquad
\T'(E_\gamma)=B e^{\alpha E_\gamma}\T(E_\gamma).\label{eq:gauge}
\end{equation}
Its product changes by the row factor $ABe^{\alpha E_i}$. Known parent fractions do not remove these three freedoms. Discrete levels, strength-shape information, and radiative-width or other absolute data constrain different combinations; each constraint must be identified explicitly.

There are additional, physical null directions. If a final spin-parity class is inaccessible from every populated initial class through the retained multipolarities, its density can be changed without affecting the primary matrix or its branching denominators. For example, dipole decay from $J_i\leq4$ does not reach $J_f\geq6$. A primary-only analysis of an allowed Gamow--Teller (GT) mixture of the two $^{90}$Rb parent spins therefore cannot determine that high-spin density. Subsequent cascade observations or higher multipolarities can change the accessible space, so this example is a statement about the specified observable.

Parity and multipolarity also produce degeneracies. In a reduced model with negative-parity initial states and the same spin weights for both final parities, the continuum kernel has the form
\begin{equation}
G(E_i,E_\gamma)=\T_{E1}(E_\gamma)\rho_+(E_f)
+\T_{M1}(E_\gamma)\rho_-(E_f).\label{eq:paritydeg}
\end{equation}
Interchanging $\rho_+\leftrightarrow\rho_-$ and $\T_{E1}\leftrightarrow\T_{M1}$ leaves the observable unchanged. Known discrete-state parities, independently determined strength components, or further population-sensitive information can break this particular symmetry. The example suffices to show that a known parent mixture does not, by itself, separate the two strength components.

Consequently, an extraction should specify whether it reports $\rho_{\rm total}$, a defined spin-parity partial density, or an effective density within Eq.~\eqref{eq:oslo}. Similarly, a parent-selected effective strength should not automatically be equated with the intrinsic sum $f_{E1}+f_{M1}$. The target parameters and the prior assumptions needed to connect them to the measured matrix are part of the physical result.

\section{Shape analysis with mixed spin-parity populations}\label{sec:shape}
\subsection{Primary transitions to matched discrete final levels}
Let $N_j(E_i)$ denote the efficiency-corrected primary yield from an initial bin to a particular discrete final level of energy $E_j$. In the mean-width description, the yield is
\begin{equation}
N_j(E_i)\propto\sum_s\frac{q_{\rm mix}(s|E_i)}{Z_s(E_i)}
\sum_{XL}\T_{XL}(E_i-E_j)S_{XL}(s;j).\label{eq:diagonal}
\end{equation}
The proportionality factor includes the number of initial decays in the bin. The discrete-level sum, rather than a continuum NLD, enters this expression. If a diagonal window contains several final levels, their number, assignments, and selection weights must be included explicitly.

For two individual final levels with the same $J_f^\pi$, the angular-momentum and parity windows coincide. Cancellation of a common population factor is the basis of the matched-state construction in the Shape method~\cite{Wiedeking2021}. Such cancellation also requires that the adopted strength description represents both final states: additional state-specific reduced strengths or substantial finite-sample fluctuations need not cancel. A narrow energy interval containing only a few populated initial levels is therefore not equivalent to a statistical average, even when the final-state assignments match.

For two $2^+$ levels and dipole transitions, negative-parity initial states with $J_i=1,2,3$ contribute through E1 and positive-parity initial states with the same spins through M1. In the simplest equal reduced-strength convention,
\begin{align}
A_{E1}(E_i)&=\sum_{J=1}^{3}\frac{q_{\rm mix}(J,-|E_i)}{Z_{J,-}(E_i)},\\
A_{M1}(E_i)&=\sum_{J=1}^{3}\frac{q_{\rm mix}(J,+|E_i)}{Z_{J,+}(E_i)}.\label{eq:Aweights}
\end{align}
Any retained spin-dependent angular factors are to be included in these sums. Equal positive- and negative-parity population probabilities alone do not imply equal $A_{E1}$ and $A_{M1}$, because the branching denominators can differ.

After division by the dipole phase-space factors, the diagonal ratio is
\begin{equation}
R_{12}(E_i)\equiv\frac{N_1(E_i)}{N_2(E_i)}
\left(\frac{E_i-E_2}{E_i-E_1}\right)^3
=
\frac{A_{E1}(E_i)f_{E1}(E_i-E_1)+A_{M1}(E_i)f_{M1}(E_i-E_1)}
{A_{E1}(E_i)f_{E1}(E_i-E_2)+A_{M1}(E_i)f_{M1}(E_i-E_2)}.
\label{eq:shaperatio}
\end{equation}
This expression preserves the same final-spin window while displaying the residual multipolarity dependence. E2 or other contributions require additional terms and different powers of $E_\gamma$.

\subsection{Intrinsic and effective strength shapes}
When one multipolarity dominates, Eq.~\eqref{eq:shaperatio} constrains the shape of that component. When
\begin{equation}
r(E_i)=A_{E1}(E_i)/A_{M1}(E_i)=r_0
\end{equation}
is constant over the analyzed rows, the ratios instead probe the fixed effective function
\begin{equation}
f_{\rm eff}(E_\gamma)=r_0 f_{E1}(E_\gamma)+f_{M1}(E_\gamma).\label{eq:feff}
\end{equation}
This equals the intrinsic sum only for the appropriate equal weighting, or for special proportional-component cases. Knowledge of a nonunit $r_0$ does not independently determine both component functions from one effective combination.

If $r(E_i)$ varies, a single scalar strength function is no longer implied by the physical model. Each row samples a different linear combination. Special cases remain possible: proportional E1 and M1 shapes, a restricted energy domain, or an underconstrained set of ratios can conceal the variation. A successful numerical sewing procedure alone cannot distinguish these cases. The relevant question is whether the available independent ratios are consistent with the claimed common function and whether their precision can resolve departures from it.

\subsection{Ratio graphs, consistency, and a fixed-separation ambiguity}\label{sec:graph}
For a candidate common positive strength $F(E_\gamma)$, put $z_v=\log F(E_{\gamma,v})$. Every corrected ratio between two gamma-energy points is an oriented edge $e=(u,v)$ with
\begin{equation}
y_e=\log R_e=z_u-z_v.
\end{equation}
Collecting the edges gives $y=Bz$, where $B$ is the graph incidence matrix. A graph with $n$ vertices and $c$ connected components has rank $n-c$. The data determine one log-strength offset per connected component only after additional information is supplied. A necessary and sufficient consistency condition for noiseless edge data is that their signed sum around every cycle vanish. In data with uncertainty, the cycle residuals and their joint covariance provide a test, without identifying the physical cause of a discrepancy.

A graph with no cycles admits a scalar representation for any set of edge values. Therefore excitation-dependent E1/M1 weighting need not produce an observable inconsistency in a sparse two-diagonal dataset. Additional final-state pairs, repeated measurements with different population selections, or time gates can create redundant constraints. Their consistency is informative only after common efficiency and statistical uncertainties are propagated.

There is also a continuous counterpart. For two fixed final energies separated by $\Delta=E_2-E_1$, all direct ratios relate gamma energies separated by the same $\Delta$. If $F$ reproduces those ratios, then
\begin{equation}
F'(x)=F(x)e^{h(x)},\qquad h(x+\Delta)=h(x)\label{eq:periodic}
\end{equation}
reproduces them as well on the measured domain. Smoothness by itself does not exclude a smooth periodic $h$. The raw ratio information therefore leaves more freedom than a single overall normalization unless additional restrictions connect the different chains of gamma energies.

This observation concerns the information in the ratios, not a claim that the practical Shape procedure is unusable. The interpolation and local-shape assumptions employed in sewing~\cite{Wiedeking2021} can select a solution. Their resolution scale and influence should be stated, particularly when a structure on the scale of $\Delta$ is discussed. The ambiguity of Eq.~\eqref{eq:periodic} also need not survive combination with an independent Oslo matrix. If that matrix has already fixed the strength shape up to $e^{\alpha E_\gamma}$, a matched pair with nonzero $\Delta$ can constrain $\alpha$ through its factor $e^{\alpha\Delta}$, provided the two analyses probe the same effective strength.

\subsection{Joint use with an Oslo matrix}
Shape and Oslo constraints derived from the same event list are statistically dependent. Some diagonal counts are contained in the primary matrix, while backgrounds, response corrections, and first-generation subtraction may affect both observables. A joint likelihood at the count or event level is preferable to multiplying likelihoods that treat these quantities as independent. If separate analyses are retained, shared resampling of the complete input data can estimate the cross-covariance. Existing Oslo uncertainty-propagation work provides a relevant statistical foundation~\cite{Midtbo2021}.

The appropriate joint target is a set of daughter properties together with population and response parameters. It is not a parent-abundance average of two already extracted strength functions. Separately obtained estimators of the same well-defined quantity may, of course, be combined with their covariance; that statistical operation is distinct from physically mixing different parent-selected effective strengths.

\section{Controlled numerical validation}\label{sec:numerical}
\subsection{A reproducible finite cascade model}\label{sec:toy}
The analytical results were tested with an exactly summed, finite Markov cascade model. No random sampling is used: each reported spectrum is the expected photon yield for the specified transition probabilities. This choice isolates the structural effects from counting noise, width fluctuations, and convergence fluctuations of Monte Carlo sampling. It supplies controlled counterexamples and implementation tests rather than an uncertainty budget for a particular experiment.

The excitation grid is $E_n=1.0+0.2n$ MeV, $n=0,\ldots,27$. Each cell has spin classes $J=0,\ldots,6$ and both parities. The $1.0$-MeV cell is an artificial absorbing manifold with these same class labels, not the $0^+$ nuclear ground state. Photons emitted below this boundary are outside the model. Initial excitation is assumed known, rather than reconstructed from an incomplete sum of the tracked photons. Keeping this boundary separate from the real discrete spectrum avoids imposing artificial nuclear assignments on the benchmark.

All models use
\begin{align}
\rho(E)&=\rho_*\exp[(E-1.0\,\mathrm{MeV})/(0.9\,\mathrm{MeV})],\label{eq:toyrho}\\
g_J(E)&=\frac{(2J+1)e^{-(J+1/2)^2/[2\sigma^2(E)]}}
{\sum_{K=0}^{6}(2K+1)e^{-(K+1/2)^2/[2\sigma^2(E)]}},\\
\frac{f_{E1}(x)}{f_*}&=1+0.15\frac{x}{\mathrm{MeV}},\qquad
\frac{f_{M1}(x)}{f_*}=0.15+0.9e^{-x/(0.8\,\mathrm{MeV})}.\label{eq:toystrength}
\end{align}
The arbitrary common scales $\rho_*$ and $f_*$ cancel in the branching probabilities. The final-state density is $\rho(E)g_J(E)p_\pi(E)$. E1 and M1 transitions obey $|J_i-J_f|\leq1$ with $0\to0$ excluded; E1 changes parity and M1 preserves it. Each allowed bin-to-bin weight is proportional to $E_\gamma^3 f_X(E_\gamma)\rho(E_f,J_f,\pi_f)\Delta E$ and is normalized separately for each initial class.

The two parent feeding functions are discrete normalized Gaussians,
\begin{equation}
B_g(E)\propto e^{-(E-3.6)^2/(2\times0.85^2)},\qquad
B_m(E)\propto e^{-(E-5.2)^2/(2\times0.70^2)},\label{eq:toyfeeding}
\end{equation}
with energies in MeV and zero direct feeding of the absorbing cell. They are normalized over the remaining cells. The integrated mixture parameter $\eta_m$ is therefore the beta fraction within this synthetic feeding domain. It is scanned from zero to one in steps of $0.1$.

Table~\ref{tab:models} defines three models. In U and S, the ground-state-labelled parent feeds only $1^-$ and the isomer-labelled parent feeds $2^-,3^-,4^-$ with equal probability. These are allowed-GT-inspired test populations, not evaluated beta strengths. Model P adds positive-parity first-forbidden-inspired fractions
\begin{equation}
F_{{\rm FF},g}(E)=0.1+\frac{0.5}{1+e^{-(E-4.4)/0.4}},\qquad
F_{{\rm FF},m}(E)=0.1,\label{eq:toyff}
\end{equation}
again with energies in MeV. The positive-parity fractions are distributed equally among $J=0,1,2$ for $g$ and $J=1,\ldots,5$ for $m$. Such assignments specify a sensitivity test; they are not a calculation of first-forbidden matrix elements in $^{90}$Rb.

\begin{table}[tb]
\centering\small
\caption{Definitions of the synthetic models. Energies entering the expressions for $\sigma^2$ and $p_+$ are in MeV. All other inputs are common and given in Eqs.~\eqref{eq:toyrho}--\eqref{eq:toyff}.}\label{tab:models}
\begin{tabularx}{\linewidth}{@{}l l l X@{}}
\toprule
Model & $\sigma^2(E)$ & Positive-parity fraction $p_+(E)$ & Direct beta population\\
\midrule
U & $2.5$ & $1/2$ & Allowed-GT-inspired windows\\
S & $1.2+0.5(E-1)$ & $1/2$ & Same as U\\
P & $1.2+0.5(E-1)$ & $[1+0.8e^{-(E-1)/1.5}]/2$ & GT windows plus Eq.~\eqref{eq:toyff}\\
\bottomrule
\end{tabularx}
\end{table}

For each model, the same daughter transition array is used for every parent fraction. Full expected cascades follow by recursion from the absorbing boundary. The implementation verifies branch normalization, primary normalization, and the photon-energy balance down to $1.0$ MeV. Both parents thus sample exactly the same synthetic daughter rather than independently generated level schemes. Figure~\ref{fig:mixture} illustrates the local parent fractions and their additional tilt under a primary-gamma energy cut.

\begin{figure}[tb]
\centering\includegraphics[width=0.87\linewidth]{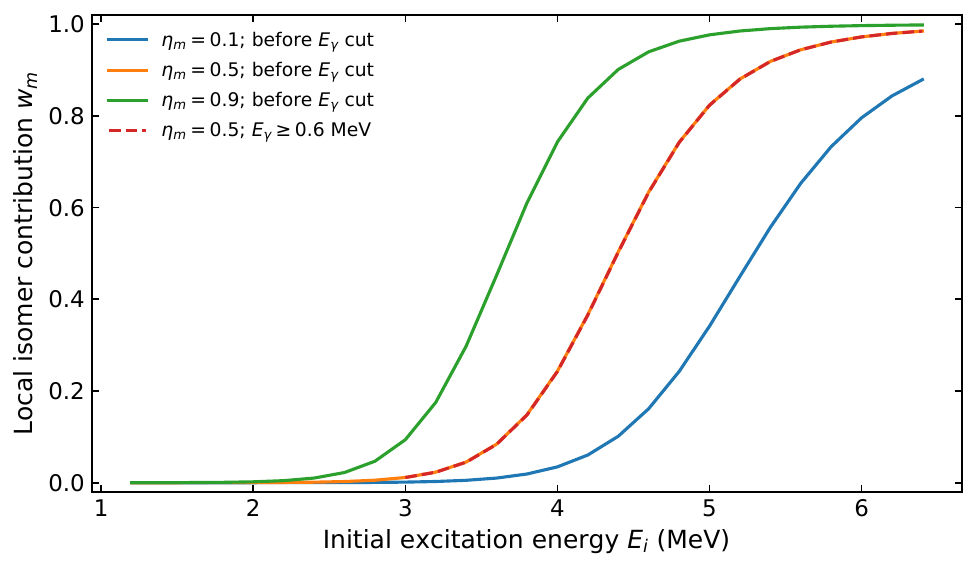}
\caption{Local isomer fraction for the synthetic feeding functions of Eq.~\eqref{eq:toyfeeding}, at integrated fractions $\eta_m=0.1$, $0.5$, and $0.9$. The additional curve conditions the $\eta_m=0.5$ model-P primary spectrum on $E_\gamma\geq0.6$ MeV. The integrated fraction, local beta fraction, and accepted primary fraction are different quantities.}\label{fig:mixture}
\end{figure}

\subsection{Primary-matrix fits and normalization conventions}\label{sec:fit}
The true primary matrices are fitted for $3.0\leq E_i\leq6.4$ MeV and $E_\gamma\geq0.6$ MeV. Every retained row is normalized over its accepted triangular support. A separate log NLD and log transmission parameter is assigned to each covered final-energy and gamma-energy bin. The objective is the equally weighted mean conditional Kullback--Leibler divergence,
\begin{equation}
\mathcal L=\frac{1}{N_i}\sum_i\KL(P_i\Vert P_i^{\rm O}),\qquad
\KL(p\Vert q)=\sum_j p_j\log\frac{p_j}{q_j}.\label{eq:klfit}
\end{equation}
This is the infinite-count conditional-multinomial objective with equal row weights. It is convex in the log parameters after removal of the gauge freedoms. A damped Newton method with an exact Hessian is used; the accepted maximum gradient is below $10^{-9}$. No counting-error bars or coverage estimates are inferred from these deterministic fits.

For comparison of the fitted functions with the known inputs, a single common gauge transformation is chosen to minimize the sum of squared log differences of both functions. Define the resulting root-mean-square errors as
\begin{equation}
\epsilon_\rho=\left[\frac{1}{N_f}\sum_f\delta r_f^2\right]^{1/2},\qquad
\epsilon_{\T}=\left[\frac{1}{N_\gamma}\sum_j\delta t_j^2\right]^{1/2}.\label{eq:gaugeerror}
\end{equation}
Here $(\delta r,\delta t)$ are the log errors after the joint least-squares removal of $(\log A,\log B,\alpha)$ in Eq.~\eqref{eq:gauge}. This is a synthetic, truth-referenced shape comparison, not an experimental normalization. In particular, a slope allocated between $\rho$ and $\T$ by this convention is not a separately measured strength-slope bias. The precise gauge projection is given in Appendix~\ref{app:compute}.

Matrix agreement is summarized by
\begin{equation}
\overline d_{\rm TV}=\frac{1}{N_i}\sum_i\TV(P_i,P_i^{\rm O}).
\end{equation}
The log-design matrix of Sec.~\ref{sec:loglinear} has 297 observed cells, 68 columns, and rank 65 for this support. The three-dimensional nullspace is the expected Oslo freedom; the corresponding row-normalized model has 47 independent shape parameters.

\subsection{Factorization with and without recovery of the total density}
Model U has energy-independent spin fractions and equal parity fractions. The accessible spin sums are initial-class-dependent constants, while the energy-dependent factor is common. These constants cancel on branch normalization. The maximum class-kernel difference is $2.2\times10^{-16}$, and every mixture is recovered to floating-point precision. This is an exact population-invariant check, including classes visited later in the cascade.

Model S retains parity balance but makes the spin distribution energy dependent. For the pure ground-state-labelled population, $J_i^\pi=1^-$, the continuum primary kernel is proportional to
\begin{equation}
\rho(E_f)\bigl[g_0(E_f)+g_1(E_f)+g_2(E_f)\bigr]
\bigl[\T_{E1}(E_\gamma)+\T_{M1}(E_\gamma)\bigr].\label{eq:Szero}
\end{equation}
The constant parity factor is immaterial. An exact Oslo representation therefore exists, but its NLD contains an energy-dependent accessible-spin fraction. The computed log-compatibility residual is $3.4\times10^{-15}$, whereas the truth-referenced joint-gauge errors are nonzero. This explicitly demonstrates the distinction between exact factorization and recovery of the total NLD.

\begin{table}[tb]
\centering\small
\caption{Selected deterministic results. $\overline d_{\rm TV}$ refers to fits of the true, accepted primary matrix. The log-shape errors use the joint truth-referenced gauge alignment of Eq.~\eqref{eq:gaugeerror}. $\overline\delta_{\rm FG}$ is the mean unthresholded photon-yield $L^1$ residual of the oracle replacement, Eq.~\eqref{eq:fgmetric}; it is not a normalized-spectrum distance. Numerical zeros for model U are below $10^{-12}$.}\label{tab:results}
\begin{tabular}{@{}lc r r r r@{}}
\toprule
Model & $\eta_m$ & $\overline d_{\rm TV}$ & $\epsilon_\rho$ & $\epsilon_{\T}$ & $\overline\delta_{\rm FG}$\\
\midrule
U & $0$--$1$ & $<10^{-12}$ & $<10^{-12}$ & $<10^{-12}$ & $<10^{-12}$\\
S & $0$ & $5.4\times10^{-12}$ & $0.0481$ & $0.0477$ & $0.0115$\\
S & $0.5$ & $0.01412$ & $0.1862$ & $0.2048$ & $0.0725$\\
S & $1$ & $0.00214$ & $0.1954$ & $0.1634$ & $0.0717$\\
P & $0$ & $0.00989$ & $0.0200$ & $0.0601$ & $0.0817$\\
P & $0.5$ & $0.01686$ & $0.1402$ & $0.1543$ & $0.1882$\\
P & $1$ & $0.00183$ & $0.1333$ & $0.1163$ & $0.0777$\\
\bottomrule
\end{tabular}
\end{table}

At intermediate mixtures, different normalized class kernels contribute with different energy-dependent weights. The standard model then generally fails the exact compatibility test. For S at $\eta_m=0.5$, the mean primary-matrix distance is only $0.0141$, although $\epsilon_\rho=0.186$ and $\epsilon_{\T}=0.205$. Model P at the same fraction has a mean distance of $0.0169$ and log-shape errors of $0.140$ and $0.154$. These numbers characterize the specified synthetic model and gauge convention; they do not establish numerical uncertainties for $^{90}$Sr. Figure~\ref{fig:factor} shows that the matrix distance need not vary monotonically with the integrated mixture fraction. Figure~\ref{fig:nld} displays the corresponding NLD shapes.

\begin{figure}[tb]
\centering\includegraphics[width=0.87\linewidth]{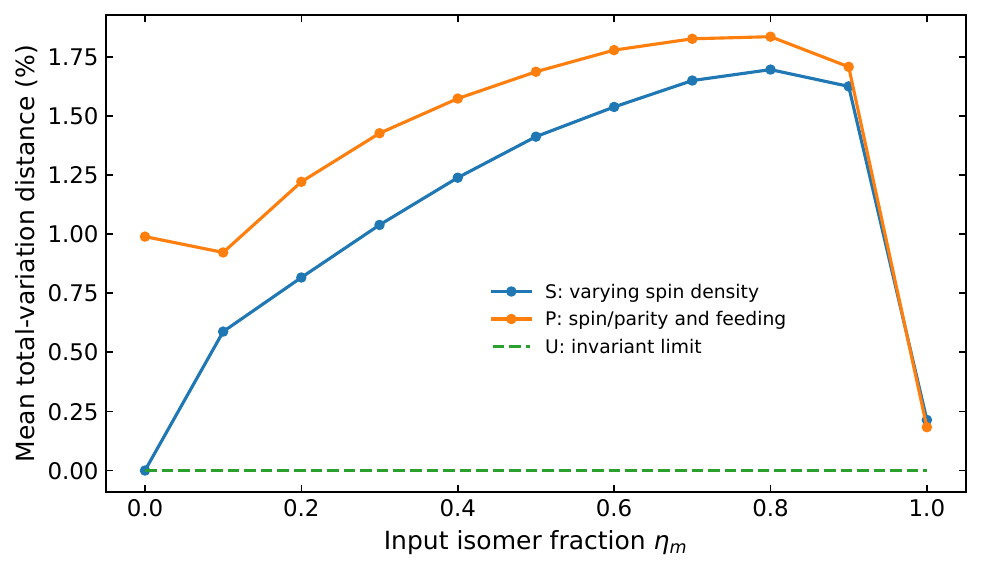}
\caption{Mean total-variation distance between the true accepted primary matrix and its best standard Oslo representation. Rows have equal weight, and the inputs other than the parent fraction are fixed within each model. Model U is factorized for every mixture to numerical precision. A small distance does not imply recovery of the total level density.}\label{fig:factor}
\end{figure}

\begin{figure}[tb]
\centering\includegraphics[width=0.87\linewidth]{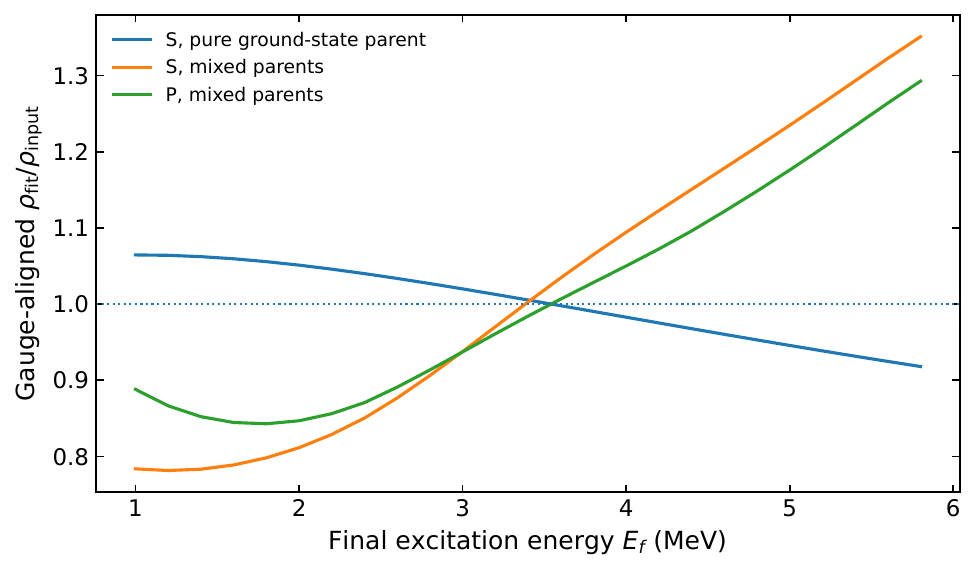}
\caption{Fitted NLD divided by the synthetic total input NLD after joint gauge alignment of the NLD and transmission errors. The two mixed-parent curves use $\eta_m=0.5$. The pure-ground-state model-S matrix has an exact Oslo factorization, yet its effective density differs from the total density. The displayed normalization uses the known synthetic truth and is not an experimental constraint.}\label{fig:nld}
\end{figure}

The small primary residual for the pure-isomer S case coexists with an appreciable effective-density discrepancy. This reinforces that goodness of fit and physical identifiability test different statements. Enlarging the range of directly fed spins can improve some aspects of coverage while changing the energy dependence of the accessible fraction; there is no general ordering in which pure or mixed populations are necessarily superior.

\subsection{Cascade-template replacement tests}
The complete cascade expectations are subjected to the oracle replacement in Eq.~\eqref{eq:oracle}. Its error is quantified without a gamma-energy cut or subsequent row normalization,
\begin{equation}
\delta_{\rm FG}(i)=\sum_j|\widehat P^{\rm or}_{ij}-P_{ij}|,\qquad
\overline\delta_{\rm FG}=\frac{1}{N_i}\sum_i\delta_{\rm FG}(i).\label{eq:fgmetric}
\end{equation}
The residual identity is verified independently to within $2\times10^{-12}$. Model U yields zero residual within floating-point precision. At $\eta_m=0.5$, the mean residual is $0.0725$ photons per initial decay for S and $0.1882$ for P. The maximum row residual for the latter is $0.2855$ (Fig.~\ref{fig:fg}). Even pure-parent samples can violate the lower-cascade template condition; the pure ground-state S case has a nonzero residual despite its exactly factorized true primary matrix.

\begin{figure}[tb]
\centering\includegraphics[width=0.87\linewidth]{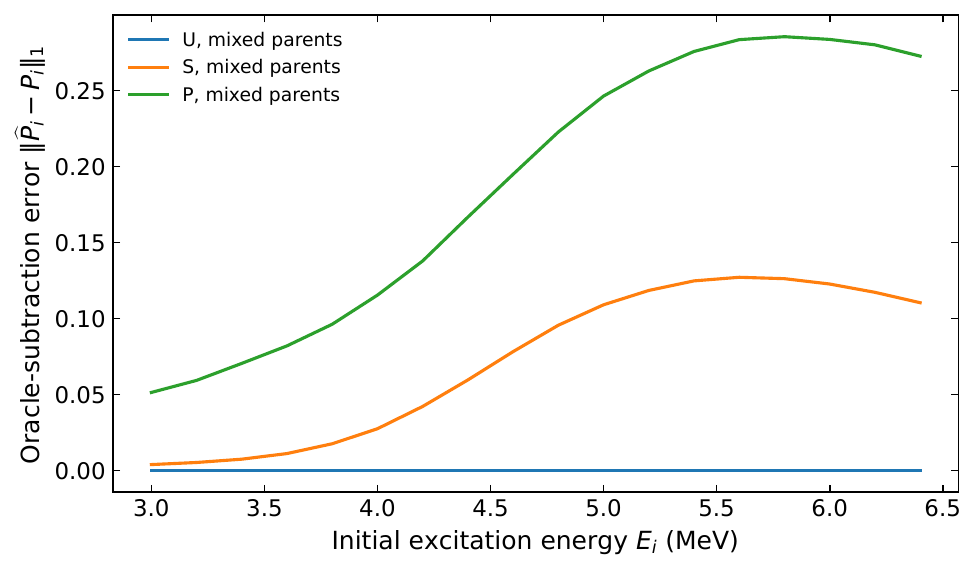}
\caption{Photon-yield $L^1$ residual of the oracle lower-cascade replacement for $\eta_m=0.5$. Exact first-destination energy probabilities are used. The curves isolate the mismatch between directly and cascade-populated lower-bin spectra; they are not results of running the iterative Oslo first-generation algorithm. All photons tracked above the artificial terminal boundary are included.}\label{fig:fg}
\end{figure}

For the pure-isomer S sample, the most affected row develops a total negative-bin weight of $0.0223$ in the signed oracle spectrum. The corresponding value for pure-isomer P is $0.00734$. Such signed outputs cannot be interpreted as normalized primary probabilities. Enforcing positivity or readjusting iterative weights may redistribute the discrepancy, but does not establish that the input template assumption has become valid.

The factorization fits in Table~\ref{tab:results} deliberately use the \emph{true} primary matrix. Fitting the signed oracle output would conflate two mechanisms and require arbitrary handling of negative bins. A detector-level analysis must assess both steps together, but their separate tests remain necessary for identifying the origin of an observed bias.

\subsection{Shape-weight and ratio-information tests}
Using the strength functions of Eq.~\eqref{eq:toystrength}, corrected ratios were calculated for final energies $0.83168$ and $1.89236$ MeV, the evaluated candidate $2^+$ energies in $^{90}$Sr~\cite{ENSDFSr90}. All strength functions and population coefficients in this calculation remain synthetic. Three cases were used: $r=1$, $r=4$, and
\begin{equation}
r(E_i)=0.25+\frac{3.75}{1+e^{-(E_i-4.8)/0.35}},\label{eq:toyr}
\end{equation}
with energies in MeV. Figure~\ref{fig:shapew} gives $\log R_{12}/\Delta$, a finite-difference slope of the corrected ratios, over $3.2\leq E_i\leq6.4$ MeV. It is a direct ratio diagnostic, not the output of a full sewing implementation. Different weights generate different apparent slopes even though the intrinsic component functions are identical in all three cases.

\begin{figure}[tb]
\centering\includegraphics[width=0.87\linewidth]{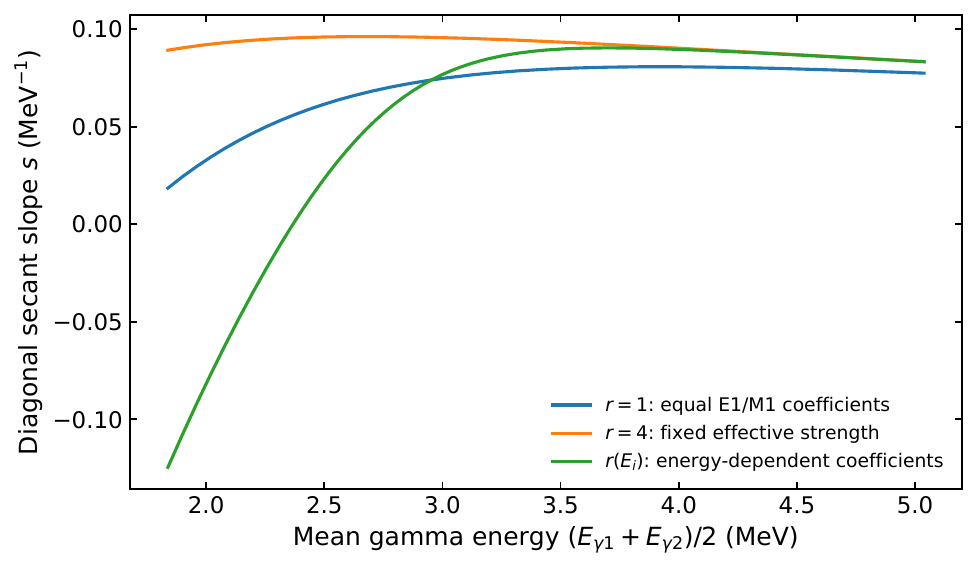}
\caption{Secant log-strength slopes inferred directly from the synthetic matched-state ratios. The horizontal coordinate is the mean of the two primary gamma energies, and $\Delta=1.06068$ MeV. Constant coefficients probe different fixed effective strengths; the variable coefficient of Eq.~\eqref{eq:toyr} samples a row-dependent combination. The figure does not assume that an unconstrained two-diagonal dataset can independently detect that variation.}\label{fig:shapew}
\end{figure}

A fixed-separation counterexample uses
\begin{equation}
F_0(x)=f_{E1}(x)+f_{M1}(x),\qquad
F_1(x)=F_0(x)\exp[0.35\sin(2\pi x/\Delta)].\label{eq:toynull}
\end{equation}
The pointwise ratio varies between approximately $0.705$ and $1.419$, whereas the two sets of fixed-separation ratios agree to a maximum relative difference of $2.7\times10^{-15}$ (Fig.~\ref{fig:null}). This example is an exact information-content test, not a proposed nuclear-strength parametrization.

\begin{figure}[tb]
\centering\includegraphics[width=0.87\linewidth]{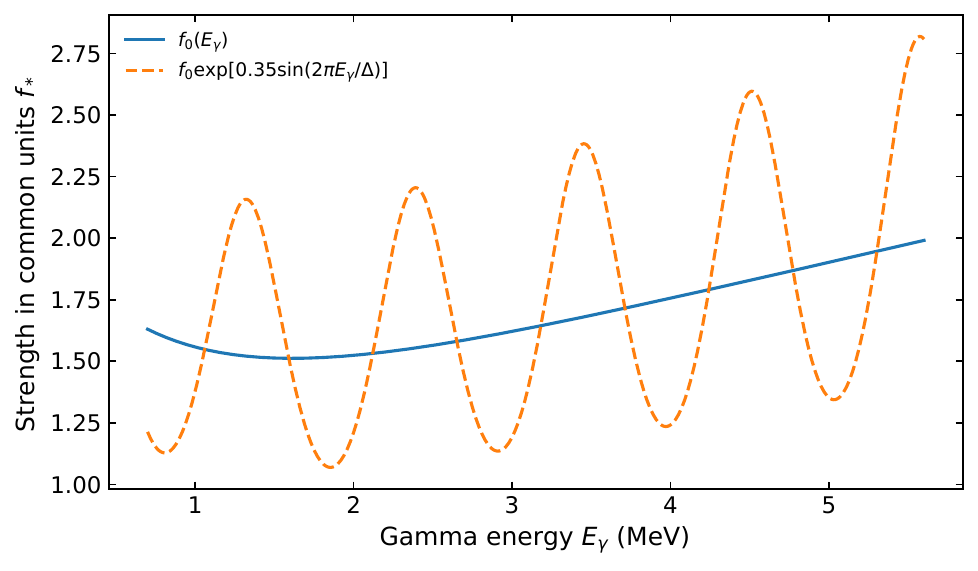}
\caption{Two positive synthetic strength functions with exactly the same ratios at separation $\Delta=1.06068$ MeV. Their logarithms differ by a periodic function. Additional energy-pair connections or shape restrictions are required to distinguish them using Shape ratios alone. Joint Oslo information can remove freedoms not present in the Oslo-compatible family.}\label{fig:null}
\end{figure}

A discrete graph test uses 21 gamma-energy nodes. Edges with index separation four produce 17 edges, four connected components, and rank 17. Adding edges of separation seven produces 31 edges, one connected component, and rank 20. The added ratios supply both inter-chain information and redundant consistency checks. These are generic grid examples, not assignments of additional $^{90}$Sr levels.

Finally, a redundant cycle was constructed using three \emph{fictitious}, identically assigned final levels at $0.8$, $1.8$, and $2.8$ MeV. Define $R(E_i;E_a,E_b;r)$ by Eq.~\eqref{eq:shaperatio}. The cycle
\begin{align}
\mathcal C={}&\log R(5.0;0.8,1.8;4)
+\log R(4.0;0.8,1.8;0.25)\nonumber\\
&-\log R(5.0;0.8,2.8;4)
\end{align}
is $-0.10293$. Setting all three coefficients to $r=1$ gives $\mathcal C=5.6\times10^{-17}$. The nonzero value demonstrates how cross-row redundancy can expose varying multipolarity weights. It does not imply that the single real candidate pair provides such a cycle.

\subsection{Scope of the numerical results}
The 33 model--fraction combinations, gauge conventions, and consistency tests are fully specified in the ancillary code. These calculations establish the behavior of the stated finite model and provide counterexamples to overly general extraction claims. They do not test a physical detector response, gamma unfolding, the full iterative first-generation algorithm, microscopic beta strengths, or Porter--Thomas fluctuations. The models also omit gamma decay below their artificial boundary. None of the reported numbers should therefore be used as a quantitative correction to $^{90}$Rb data without an experiment-specific forward calculation.

A realistic extension can retain the same hierarchy: true-primary compatibility, cascade-template consistency, and detector-level recovery. A level-realization Monte Carlo calculation should use the same daughter realization for both parent components, then repeat across realizations. Counting fluctuations and model misspecification should be varied independently. Increasing counts can reduce statistical uncertainty while leaving the structural biases identified here unchanged.

\section{The \texorpdfstring{$\rb\rightarrow\sr$}{90Rb to 90Sr} case and time-dependent mixtures}\label{sec:rb}
\subsection{Evaluated constraints and their role}
The $^{90}$Rb system provides two long-lived parents feeding one even-even daughter. Table~\ref{tab:nuclear} summarizes the evaluated quantities used here. Both parent states have negative parity, but their spins differ. Their beta-decay endpoints lie below the adopted $^{90}$Sr neutron separation energy, making a below-neutron-threshold radiative-decay treatment relevant to this case~\cite{ENSDFRb90,ENSDFSr90}. This kinematic observation does not establish statistical decay in every excitation interval; the density of populated states and the influence of discrete transitions must still be assessed.

\begin{table}[tb]
\centering\small
\caption{Evaluated inputs relevant to $^{90}$Rb parent mixing and the candidate $^{90}$Sr Shape pair. Parent data are from Ref.~\cite{ENSDFRb90} and daughter data from Ref.~\cite{ENSDFSr90}, both based on the March 2020 evaluation. Uncertainties in parentheses follow those adopted datasets. The energy-release and separation-energy values are quoted at the precision needed for the threshold comparison.}\label{tab:nuclear}
\begin{tabular}{@{}l l l@{}}
\toprule
System or quantity & Adopted property & Relevance\\
\midrule
$^{90}$Rb ground state & $J^\pi=0^-$; $T_{1/2}=158(5)$ s & Low-spin beta population\\
$^{90m}$Rb & $E_m=106.90(3)$ keV; $J^\pi=3^-$ & Higher-spin parent\\
$^{90m}$Rb half-life & $258(4)$ s & Temporal separation\\
$^{90m}$Rb decay branches & $\beta^-:97.5(4)\%$; IT: $2.5(4)\%$ & Feeding of ground-state activity\\
$Q_\beta(^{90}\mathrm{Rb}_{g})$ & $6.584(7)$ MeV & Ground-parent endpoint\\
$S_n(^{90}\mathrm{Sr})$ & Approximately $7.810$ MeV & Above both parent endpoints\\
$^{90}$Sr first candidate & $831.68(4)$ keV, $2_1^+$ & Matched final-state pair\\
$^{90}$Sr second candidate & $1892.36(4)$ keV, $2_2^+$ & $\Delta=1.06068$ MeV\\
\bottomrule
\end{tabular}
\end{table}

In the allowed Gamow--Teller limit, the rank-one spin operator gives
\begin{equation}
0^-\xrightarrow{\rm GT}1^-,\qquad
3^-\xrightarrow{\rm GT}2^-,3^-,4^-.\label{eq:GT}
\end{equation}
A $0^-$ daughter state is not part of the first GT window: a rank-one operator cannot connect $J=0$ to $J=0$. An allowed Fermi transition would instead have rank zero and require the appropriate isospin structure and an energetically accessible state. It should not be inserted into a GT population model solely from the equality of spin and parity.

First-forbidden transitions can populate opposite-parity states and therefore alter both the spin coverage and the E1/M1 coefficients in Eq.~\eqref{eq:Aweights}. The evaluated ground-parent decay scheme includes significant feeding to low-lying positive-parity levels, together with qualifications concerning the inferred feeding distribution~\cite{ENSDFBeta90}. Those low-energy branches establish that a purely allowed description of the entire decay is insufficient. They do not determine the first-forbidden fraction in the higher excitation window used for a statistical analysis. Accordingly, the two functions $F_{{\rm FF},a}(E_i)$ must be constrained separately or treated as population-model uncertainties.

The pair of evaluated $2^+$ levels is a suitable candidate for a matched-final-state test, subject to the visibility and purity of the primary diagonals. For negative-parity initial feeding it probes E1 transitions from $J_i=1,2,3$; positive-parity feeding contributes through M1 over the same spin set. In the GT limit the isomer-fed $J_i=4^-$ class cannot decay directly to either $2^+$ level by a dipole transition, although it contributes to other parts of the primary spectrum. The Shape and full-matrix observables consequently sample different subsets even before first-forbidden contributions are added.

The synthetic models in Sec.~\ref{sec:numerical} use these spin windows and, in the separate Shape test, the two final-state energies. Their continuum densities, strength functions, and beta-feeding profiles are otherwise constructed inputs. The calculation supports feasibility criteria and failure mechanisms, not a preferred empirical $^{90}$Sr strength curve or a numerical correction for an existing measurement.

\subsection{Activities, internal transitions, and observation windows}\label{sec:time}
Distinct half-lives offer information that is lost when all events are integrated. For a simple post-collection decay interval with no further implantation, let $N_g(0)$ and $N_m(0)$ be the parent inventories. With the adopted small isomeric internal-transition branch $b_{\rm IT}$,
\begin{align}
\dot N_m&=-\lambda_mN_m,\\
\dot N_g&=-\lambda_gN_g+b_{\rm IT}\lambda_mN_m.
\end{align}
Assuming the ground state decays by beta emission, the relevant activities are
\begin{equation}
A_g(t)=\lambda_gN_g(t),\qquad
A_m(t)=(1-b_{\rm IT})\lambda_mN_m(t).\label{eq:activities}
\end{equation}
Here an isomer that undergoes an internal transition and later beta decays contributes to the ground-parent beta spectrum. The solutions and integrals are given in Appendix~\ref{app:time}. Continuous implantation requires the corresponding source terms and cannot be reduced to the post-collection formula without accounting for the collection history.

Let $C_a$ be the expected, unnormalized detector-bin distribution per beta decay from parent $a$, including its feeding and fixed response. For observation window $\ell$,
\begin{equation}
C_\ell=h_{\ell g}C_g+h_{\ell m}C_m,
\qquad h_{\ell a}=\int_{t_\ell}^{t_{\ell+1}} A_a(t)L(t)\dd t,\label{eq:timeC}
\end{equation}
where $L(t)$ is the live-time fraction. Backgrounds and parent-dependent efficiencies are added or folded into the relevant response. When the inventories and timing parameters are constrained, the two column vectors of $H=(h_{\ell a})$ describe the available temporal separation. If the inventories are uncertain they become nuisance parameters in a joint fit; knowledge of the two half-lives alone does not fix $H$ completely.

Rank two is necessary for separating two arbitrary component spectra through Eq.~\eqref{eq:timeC}. Near-collinear columns, however, give a poorly conditioned separation and amplify noise. Row normalization before the fit discards amplitude information and introduces the excitation-dependent fractions of Eq.~\eqref{eq:localw}; the simple linear deconvolution should therefore use unnormalized counts. A joint fit need not assign a unique parent to every event.

\subsection{A temporal conditioning benchmark}
For illustration, equal initial inventories, unit live time and efficiency, and six equal-width observation windows were used with the evaluated central half-lives. Each column of $H$ was normalized to unit sum before calculating its Euclidean condition number. This convention measures similarity of the temporal shapes rather than their total amplitudes. Including $b_{\rm IT}=0.025$ gives
\begin{equation}
\begin{array}{c|rrrr}
T_{\rm obs}\ (\mathrm{s}) & 60 & 300 & 600 & 1200\\ \hline
\kappa(H_{\rm col}) & 72.18 & 16.13 & 10.54 & 9.62
\end{array}\label{eq:timevalues}
\end{equation}
for the specified six-window design. Figure~\ref{fig:time} also shows the result without the internal-transition branch. The large short-window condition number is consistent with the two exponential activities being nearly proportional over a brief interval.

\begin{figure}[tb]
\centering\includegraphics[width=0.87\linewidth]{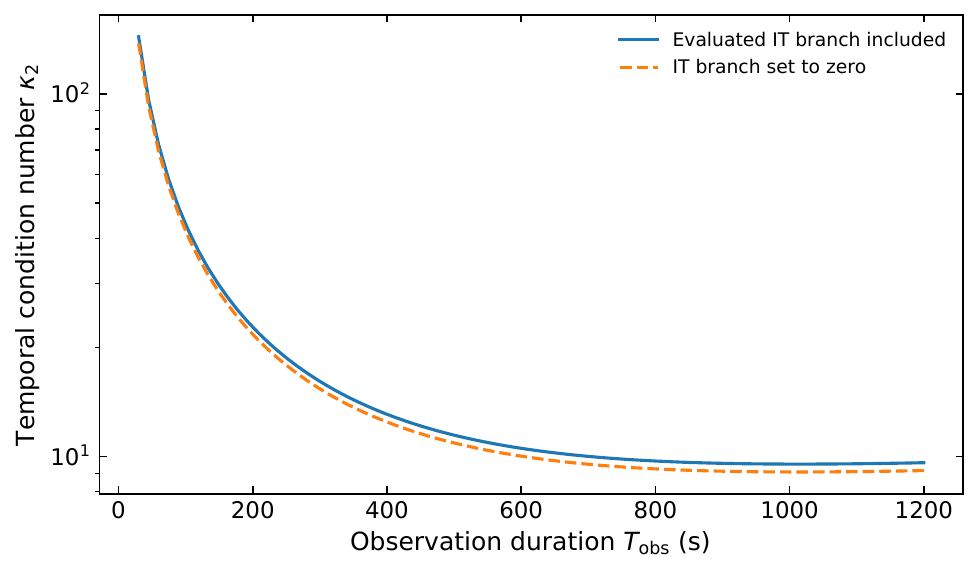}
\caption{Condition number of the column-normalized temporal mixing matrix for six equal observation windows and equal initial $^{90}$Rb ground-state and isomer inventories. Central evaluated half-lives are used. Curves with and without the internal-transition branch compare timing-shape distinguishability only; they do not include counting noise, half-life uncertainty, backgrounds, or collection history.}\label{fig:time}
\end{figure}

A lower unweighted condition number is not, by itself, an optimal acquisition criterion. The number of retained events, late-time backgrounds, the widths of the bins, and covariance of the timing parameters also matter. With a fixed number of ever wider windows, an indefinitely long observation interval can eventually place almost all activity in the first bin. The finite-duration example therefore provides no general prescription to maximize observation time. Its role is to demonstrate how actual half-lives can be incorporated into an identifiability assessment before attempting an unconstrained spectral separation.

\section{Implications for experimental analysis}\label{sec:discussion}
\subsection{What the available observations can establish}
An unresolved-parent sample can be analyzed as a single observed ensemble. Whether a standard extraction has a defensible physical interpretation then depends on the tested decay kernels, the first-generation assumption, and the definition of the inferred quantities. Known global fractions supply useful information but do not replace energy-dependent feeding or spin-parity constraints. Conversely, unknown local fractions need not invalidate a standard analysis in a verified population-invariant regime.

Table~\ref{tab:scope} summarizes the information provided by several data combinations. The distinctions are structural: adding a strong prior can produce a narrow parameter interval without making the corresponding quantity identifiable from the data alone. In practice, sensitivity to plausible population and response alternatives should be reported separately from conditional statistical errors.

\begin{table}[tb]
\centering\small
\caption{Information and remaining conditions for representative observation sets. “Constraint” does not imply absolute normalization or uniqueness outside the specified model.}\label{tab:scope}
\begin{tabularx}{\linewidth}{@{}>{\raggedright\arraybackslash}p{0.25\linewidth} Y Y@{}}
\toprule
Observations & Direct constraint & Essential remaining issue\\
\midrule
Mixed primary matrix and integrated parent fraction & Observed ensemble and compatibility with a factorized family & Local feeding, accessible level space, and Oslo normalization\\
\addlinespace
Known $q_{\rm mix}(E_i,J^\pi)$ and a primary matrix & Population-aware branching model & Unobserved channels and residual parameter null directions\\
\addlinespace
Matched-state Shape ratios & Selected multipolarity-weighted strength ratios & Weight stability, ratio connectivity, and normalization\\
\addlinespace
Several time windows or tagged population samples & Independent parent-mixture combinations & Temporal rank, statistical conditioning, and response stability\\
\addlinespace
Complete cascade and response model plus external anchors & Joint constraints on specified daughter properties & Model misspecification and independent closure validation\\
\bottomrule
\end{tabularx}
\end{table}

\subsection{A population-aware forward likelihood}
A consistent extension of standard analysis can use shared daughter parameters $\theta=\{\rho(E,J,\pi),f_{XL}(E_\gamma),\ldots\}$ and parent or timing nuisance parameters $\psi$. Schematically, the expected counts are
\begin{equation}
\mu_{\ell k}(\theta,\psi)=
\sum_a h_{\ell a}(\psi)
\sum_\nu B_a(\nu|\psi)
\mathcal R_{\ell k}\!\left[\mathcal C_\nu(\theta);\psi\right]
+\mu^{\rm bg}_{\ell k}(\psi),\label{eq:joint}
\end{equation}
where $\mathcal C_\nu$ is the full event-level cascade distribution and $\mathcal R$ maps it to the chosen detector observable. Absolute amplitudes, multiplicities, and sum-energy migration can then be retained. This is the natural mixed-parent extension of population-dependent forward modeling~\cite{Zeiser2019}.

Independent Poisson factors are justified for disjoint event categories, not automatically for gamma-hit bins that can receive several entries from one cascade. An event likelihood, a compound-count description, or a validated covariance model is needed in the latter case. Shape diagonals and the primary matrix should be derived jointly from this model or compared with their cross-covariance. Parameter identifiability can be assessed from profile likelihoods or the rank of the appropriately weighted local sensitivity matrix after known gauge directions have been removed.

The data do not require that both parents have identical \emph{effective} Oslo outputs. They require that one coherent daughter model explain all parent selections, time windows, and retained observables. Failure can arise from the daughter-strength parametrization, beta feeding, inadequate statistical averaging, or the response model. Attributing every discrepancy to isomer contamination would conceal these alternatives.

\subsection{Validation and reporting boundaries}
A useful validation sequence begins with the true primary distribution, proceeds to complete cascades and the actual first-generation implementation, and finally includes the measured detector response and acquisition conditions. Simulating and fitting with the same model checks implementation; generating with alternative feeding or strength models tests robustness. Event statistics should be held fixed when comparing mixtures, so that a change in precision is not confused with a change in population bias.

If high-purity or tagged parent samples exist, a blind mixture constructed at the event level tests the handling of known proportions, windows, and correlations. Agreement with separate analyses is valuable but does not establish absolute correctness, since all analyses can share a population or normalization bias. Synthetic truth tests and external anchors remain complementary.

A report should identify the parent-fraction definition, accepted energy support, nuclear-data constraints, treatment of first-forbidden feeding, normalization anchors, and the distinction between total, partial, and effective quantities. For Shape analysis it should also specify the discrete levels, contamination treatment, assumed multipolarities, and the interpolation or smoothness information used to connect ratio chains. These definitions make the result reusable for a subsequent reaction calculation or a comparison with another population mechanism.

\section{Conclusions}\label{sec:conclusions}
Unresolved beta-decaying parent states feeding a common daughter can be treated in a single forward model. The mixed primary spectrum is a population average of normalized branching probabilities. Integrated parent fractions become excitation-dependent through beta feeding and are further modified by event and gamma-energy selection. Averaging unnormalized decay kernels and normalizing only afterwards generally changes the physical ensemble.

Population invariance, conventional Oslo factorization, and identification of the total level density are separate properties. A positive-matrix log-linear test establishes exact factorization compatibility, while an explicit accessible-spin example shows that perfect compatibility can coexist with an effective NLD different from the total density. Normalization freedoms and unobserved spin-parity sectors remain even for a known mixture. A separate cascade-replacement identity isolates the population assumption entering first-generation subtraction.

For the Shape method, identical final-state spins and parities align selection windows but do not automatically remove E1/M1 weighting. Constant coefficients define a fixed effective strength; excitation-dependent coefficients require additional scrutiny. Ratio-graph connectivity determines the information available before interpolation or other shape assumptions are imposed. In particular, one fixed final-state separation admits a periodic log-strength ambiguity in the raw ratios. Additional independent edges can both connect energy chains and test consistency.

The deterministic benchmarks verify these statements and demonstrate that small primary-matrix residuals can accompany appreciable errors in truth-referenced strength and density shapes. Evaluated $^{90}$Rb and $^{90}$Sr properties define a relevant mixed-parent application and a candidate matched $2^+$ pair. The half-life difference provides an additional temporal observable, subject to conditioning and counting precision. Quantitative extraction from an actual $^{90}$Rb measurement requires its beta-feeding constraints, response, and full analysis pipeline; the present results supply the tests and definitions needed to interpret that extraction.

\section*{Data and code availability}
The source package contains the executable Python benchmark implementation, machine-readable input definitions and results, and the arrays used to generate all figures in the \texttt{anc/} directory. The numerical results can be regenerated without experimental data or external nuclear-reaction software. Evaluated nuclear quantities are identified in Table~\ref{tab:nuclear} and Refs.~\cite{ENSDFRb90,ENSDFSr90,ENSDFBeta90}.

\appendix
\section{Additional normalization and ratio identities}\label{app:identities}
\subsection{Mixing conditioned spectra}
Let $C_a(i,j)=N_a(i)P_a(j|i)$ denote the expected true primary counts before a gamma cut. The mixed counts satisfy $C=\sum_aC_a$. Normalization over a restricted set $\W_i$ gives
\begin{equation}
\frac{C(i,j)}{\sum_{k\in\W_i}C(i,k)}
=\sum_a
\frac{N_a(i)A_a(i)}{\sum_b N_b(i)A_b(i)}
\frac{P_a(j|i)}{A_a(i)},
\end{equation}
which proves Eq.~\eqref{eq:acceptedw}. A positive, known energy-dependent efficiency can be included in $A_a$ and the accepted conditional spectrum. Energy migration requires summation over the contributing true initial bins as well, and therefore cannot generally be represented by a fixed true-row weight.

\subsection{The log-design gauge}
The transformation
\begin{equation}
r_f\mapsto r_f+a+\alpha E_f,\quad
 t_j\mapsto t_j+b+\alpha E_{\gamma,j},\quad
 c_i\mapsto c_i-a-b-\alpha E_i
\end{equation}
leaves $X\theta$ invariant because $E_f+E_\gamma=E_i$. These provide three null vectors on an adequately connected support. Restricted or disconnected supports may have additional null vectors, so the rank should be evaluated for the actual energy cuts rather than presumed. The rank stated in Sec.~\ref{sec:fit} is computed from the specified 297-cell support using a relative singular-value threshold of $10^{-12}$.

\subsection{Ratio consistency and nullspace}
Choose one vertex in each connected component of the ratio graph. If every cycle sum vanishes, assign its log-strength value arbitrarily and integrate edge differences along a path to every other vertex. Two alternative paths form a cycle, so the result is path independent. This constructs $z$ with $Bz=y$ and proves sufficiency of the cycle condition. Necessity follows from telescoping the vertex differences. The remaining freedom is a constant on each component.

For the continuum fixed-separation problem, if two positive functions have identical ratios, the difference of their logarithms satisfies $h(x+\Delta)-h(x)=0$ wherever both arguments are measured. Conversely, that equality immediately preserves every ratio. On a finite domain the freedom is a function on one set of representatives of the $\Delta$-separated chains, propagated only over the observed range. Assuming a low-dimensional smooth parametrization may remove this freedom, but the resulting information then includes that parametrization.

\section{Time-dependent activities and window integrals}\label{app:time}
For distinct $\lambda_g$ and $\lambda_m$, the post-collection populations are
\begin{align}
N_m(t)&=N_m(0)e^{-\lambda_mt},\\
N_g(t)&=N_g(0)e^{-\lambda_gt}
+b_{\rm IT}\lambda_mN_m(0)
\frac{e^{-\lambda_mt}-e^{-\lambda_gt}}{\lambda_g-\lambda_m}.\label{eq:bateman}
\end{align}
For a window $[u,v]$ with unit live time, define $D_a(u,v)=e^{-\lambda_au}-e^{-\lambda_av}$. The beta counts contributed by each parent state are
\begin{align}
h_g(u,v)&=N_g(0)D_g
+\frac{b_{\rm IT}\lambda_mN_m(0)}{\lambda_g-\lambda_m}
\left[\frac{\lambda_g}{\lambda_m}D_m-D_g\right],\label{eq:hg}\\
h_m(u,v)&=(1-b_{\rm IT})N_m(0)D_m.\label{eq:hm}
\end{align}
These expressions separate the state from which beta decay occurs from the state initially implanted. The near-degenerate limit should be evaluated by a stable limiting expression rather than direct subtraction of nearly equal exponentials. Nonuniform live time requires the integrals in Eq.~\eqref{eq:timeC}; an implantation source can be convolved with Eq.~\eqref{eq:bateman}.

For Fig.~\ref{fig:time}, $N_g(0)=N_m(0)=1/2$, $\lambda_g=\log2/(158\,\mathrm{s})$, and $\lambda_m=\log2/(258\,\mathrm{s})$. Six equally spaced windows partition $[0,T_{\rm obs}]$. With $\widetilde H_{\ell a}=h_{\ell a}/\sum_k h_{ka}$, the plotted quantity is the ratio of the largest and smallest singular values of $\widetilde H$. The branch-free comparison sets $b_{\rm IT}=0$ and consequently makes the isomer beta branch unity. This column normalization is diagnostic only and is not applied to measured counts before inference.

\section{Computational implementation and reproduction}\label{app:compute}
\subsection{Cascade recursion and numerical checks}
The ancillary program stores the normalized transition tensor $b_{is,ft}$ for $f<i$. Its exact expected primary and cascade spectra obey
\begin{align}
P_{is,j}&=\sum_{f,t}b_{is,ft}\,\delta_{j,i-f},\\
H_{is,j}&=P_{is,j}+\sum_{f,t}b_{is,ft}H_{ft,j},
\end{align}
with $H_{0s,j}=0$ at the absorbing manifold. Since the energy grid is ordered, a single ascending pass gives the complete expectation without sampling. The checks require branch and primary sums to agree with unity within $10^{-12}$ and the total gamma energy to equal $E_i-1.0$ MeV within $10^{-11}$ MeV. The direct and independently evaluated forms of Eq.~\eqref{eq:oracle} agree within $2\times10^{-12}$.

In the allowed-only model, the program verifies that first transitions from populated $J_i\leq4$ have zero weight into $J_f\geq6$. It also checks the two-bin normalization counterexample, the total-variation identity, the Shape graph ranks, and the periodic-ratio invariance. No random seed is required because none of these computations uses random draws.

\subsection{Fit objective and gauge alignment}
For a cell $(i,j)$ the fitting score is $u_{ij}=r_{i-j}+t_j$ and the conditional prediction is the row softmax of $u$. The gradient of Eq.~\eqref{eq:klfit} is the mean difference between predicted and input sufficient statistics; its Hessian is their predicted covariance. Three independent coordinates are fixed during optimization. A small numerical diagonal stabilization of $10^{-14}$ is used when solving the Newton system, together with a backtracking line search. The supplied result file records the final gradients and iteration counts.

Let $v$ concatenate the fitted log differences from the synthetic input NLD and transmission before comparison. Define a matrix $G$ with three columns: an NLD-only constant, a transmission-only constant, and the combined energy vector $(E_f,E_\gamma)$. The plotted and tabulated shape errors are
\begin{equation}
v_{\perp}=v-G(G^{\mathsf T}G)^{-1}G^{\mathsf T}v.
\end{equation}
The implementation uses a least-squares solve rather than explicitly forming the inverse. All covered bins enter with equal weight. For pure-ground-state model S, part of the energy dependence of the accessible density can consequently be assigned to a compensating transmission slope. The nonzero individual $\epsilon_{\T}$ in that case must be read with this convention; the invariant statement is that the pair of fitted functions cannot be made identical to the total-density input and intrinsic transmission simultaneously by the three standard gauge changes.

\subsection{Reproduction commands and archived outputs}
From the source-package root, the commands are
\begin{verbatim}
python anc/benchmarks.py --compute
python anc/benchmarks.py --plots
pdflatex main.tex
bibtex main
pdflatex main.tex
pdflatex main.tex
\end{verbatim}
The archived computation used Python 3.13.5, NumPy 2.3.5, SciPy 1.17.0, and Matplotlib 3.10.8. The code requires only these Python packages and standard-library modules. The precomputed vector figures and bibliography are included, so Python is not needed to compile the manuscript. The machine-readable outputs comprise \texttt{results.json}, \texttt{fit\_metrics.csv}, and \texttt{benchmark\_arrays.npz}. They contain all 33 primary-fit cases and the independent normalization, graph, Shape, and temporal checks.

The manuscript is compiled with PDFLaTeX using standard TeX Live packages. No shell-escape graphics conversion, external font file, network resource, or downloaded dataset is required during compilation. Evaluated nuclear values used by the tests are explicit constants in the code and are documented in the main text.

\bibliographystyle{unsrtnat}
\bibliography{references}

@article{Schiller2000,
  author = {Schiller, A. and Bergholt, L. and Guttormsen, M. and Melby, E. and Rekstad, J. and Siem, S.},
  title = {Extraction of level density and {$\gamma$} strength function from primary {$\gamma$} spectra},
  journal = {Nucl. Instrum. Methods Phys. Res. A}, volume = {447}, pages = {498--511}, year = {2000},
  doi = {10.1016/S0168-9002(99)01187-0}, note = {arXiv:nucl-ex/9910009}
}

@article{Guttormsen1987,
  author = {Guttormsen, M. and Rams{\o}y, T. and Rekstad, J.},
  title = {The first generation of {$\gamma$}-rays from hot nuclei},
  journal = {Nucl. Instrum. Methods Phys. Res. A}, volume = {255}, pages = {518--523}, year = {1987},
  doi = {10.1016/0168-9002(87)91221-6}
}

@article{Spyrou2014,
  author = {Spyrou, A. and Liddick, S. N. and Larsen, A. C. and Guttormsen, M. and Cooper, K. and Dombos, A. C. and others},
  title = {Novel technique for constraining {$r$}-process {$(n,\gamma)$} reaction rates},
  journal = {Phys. Rev. Lett.}, volume = {113}, pages = {232502}, year = {2014},
  doi = {10.1103/PhysRevLett.113.232502}
}

@article{Larsen2011,
  author = {Larsen, A. C. and Guttormsen, M. and Krti{\v c}ka, M. and B{\v e}t{\'a}k, E. and B{\"u}rger, A. and G{\"o}rgen, A. and others},
  title = {Analysis of possible systematic errors in the {Oslo} method},
  journal = {Phys. Rev. C}, volume = {83}, pages = {034315}, year = {2011},
  doi = {10.1103/PhysRevC.83.034315}, note = {arXiv:1211.6264}
}

@article{Zeiser2019,
  author = {Zeiser, F. and Tveten, G. M. and Potel, G. and Larsen, A. C. and Guttormsen, M. and Laplace, T. A. and others},
  title = {Restricted spin-range correction in the {Oslo} method: The example of nuclear level density and {$\gamma$}-ray strength function from {$^{239}\mathrm{Pu}(d,p\gamma)^{240}\mathrm{Pu}$}},
  journal = {Phys. Rev. C}, volume = {100}, pages = {024305}, year = {2019},
  doi = {10.1103/PhysRevC.100.024305}, note = {arXiv:1904.02932}
}

@article{Wiedeking2021,
  author = {Wiedeking, M. and Guttormsen, M. and Larsen, A. C. and Zeiser, F. and G{\"o}rgen, A. and Liddick, S. N. and M{\"u}cher, D. and Siem, S. and Spyrou, A.},
  title = {Independent normalization for {$\gamma$}-ray strength functions: The {Shape} method},
  journal = {Phys. Rev. C}, volume = {104}, pages = {014311}, year = {2021},
  doi = {10.1103/PhysRevC.104.014311}, note = {arXiv:2010.15696}
}

@article{Muecher2023,
  author = {M{\"u}cher, D. and Spyrou, A. and Wiedeking, M. and Guttormsen, M. and Larsen, A. C. and Zeiser, F. and others},
  title = {Extracting model-independent nuclear level densities away from stability},
  journal = {Phys. Rev. C}, volume = {107}, pages = {L011602}, year = {2023},
  doi = {10.1103/PhysRevC.107.L011602}
}

@article{Ronning2026,
  author = {Ronning, E. K. and Richard, A. L. and Liddick, S. N. and Spyrou, A. and others},
  title = {Magnetic character of the low-energy enhancement in {$^{70}$Zn}},
  journal = {Nature}, volume = {655}, pages = {875--878}, year = {2026},
  doi = {10.1038/s41586-026-10758-3}
}

@article{Midtbo2021,
  author = {Midtb{\o}, J. E. and Zeiser, F. and Lima, E. and Larsen, A.-C. and Tveten, G. M. and Guttormsen, M. and Bello Garrote, F. L. and Kvellestad, A. and Renstr{\o}m, T.},
  title = {A new software implementation of the {Oslo} method with rigorous statistical uncertainty propagation},
  journal = {Comput. Phys. Commun.}, volume = {262}, pages = {107795}, year = {2021},
  doi = {10.1016/j.cpc.2020.107795}, note = {arXiv:1904.13248}
}

@article{Kirsch2018,
  author = {Kirsch, L. E. and Bernstein, L. A.},
  title = {{RAINIER}: A simulation tool for distributions of excited nuclear states and cascade fluctuations},
  journal = {Nucl. Instrum. Methods Phys. Res. A}, volume = {892}, pages = {30--40}, year = {2018},
  doi = {10.1016/j.nima.2018.02.096}, note = {arXiv:1709.04006}
}

@article{Porter1956,
  author = {Porter, C. E. and Thomas, R. G.},
  title = {Fluctuations of nuclear reaction widths},
  journal = {Phys. Rev.}, volume = {104}, pages = {483--491}, year = {1956},
  doi = {10.1103/PhysRev.104.483}
}

@misc{ENSDFRb90,
  author = {{National Nuclear Data Center}},
  title = {{ENSDF}: {$^{90}$Rb} adopted levels, gammas},
  howpublished = {Evaluation by S. K. Basu and E. A. McCutchan, March 2020; Nuclear Data Sheets 165, 1 (2020)},
  year = {2020},
  url = {https://www.nndc.bnl.gov/ensdf/getadopted.jsp?nuc=90Rb},
  note = {Accessed 20 September 2026}
}

@misc{ENSDFSr90,
  author = {{National Nuclear Data Center}},
  title = {{ENSDF}: {$^{90}$Sr} adopted levels, gammas},
  howpublished = {Evaluation by S. K. Basu and E. A. McCutchan, March 2020; Nuclear Data Sheets 165, 1 (2020)},
  year = {2020},
  url = {https://www.nndc.bnl.gov/ensnds/90/Sr/adopted.pdf},
  note = {Accessed 20 September 2026}
}

@misc{ENSDFBeta90,
  author = {{National Nuclear Data Center}},
  title = {{ENSDF}: {$^{90}$Rb} {$\beta^-$} decay (158 s), daughter {$^{90}$Sr}},
  howpublished = {Evaluated decay dataset, NuDat 3},
  year = {2020},
  url = {https://www.nndc.bnl.gov/nudat3/getdecaydataset.jsp?dsid=90rb+bM+decay+(158+s)&nucleus=90SR},
  note = {Accessed 20 September 2026}
}
\end{document}